\documentclass[final,runningheads,a4paper]{llncs}

\usepackage[T1]{fontenc}
\usepackage[latin1]{inputenc}
\usepackage{amsmath,amssymb}
\usepackage{wasysym}
\usepackage{enumitem}
\usepackage{framed}
\usepackage{alltt}
\usepackage{mathrsfs}

\newtheorem{alg}{Algorithm}

\newcommand{\C}{\mathscr{C}}
\newcommand{\B}{\mathscr{B}}
\newcommand{\R}{\mathbb{R}}
\newcommand{\Q}{\mathbb{Q}}
\newcommand{\N}{\mathbb{N}}
\newcommand{\fri}[1]{#1^\blacklozenge}

\newcommand{\bvp}[2]{\boxed{\begin{array}{l}#1\\#2\end{array}}}

\DeclareMathOperator{\Ker}{Ker}
\DeclareMathOperator{\img}{Im}

\DeclareMathOperator{\codim}{codim}
\DeclareMathOperator{\ord}{ord}
\DeclareMathOperator{\Ll}{span}

\newcommand{\Bo}{\B^{\perp}}
\newcommand{\Co}{\C^{\perp}}

\newcommand{\dx}{\,\mathrm{d}}

\newcommand{\E}{\textbf{\scshape\texttt E}} 
\newcommand{\A}{\operatorname{A}}
\newcommand{\D}{\operatorname{D}}
\newcommand{\BC}{\operatorname{BC}}
\newcommand{\ES}{\operatorname{ES}}
\newcommand{\BP}{\operatorname{BP}}
\newcommand{\GBP}{\operatorname{GBP}}
\newcommand{\pack}{\texttt{IntDiffOp}}
\newcommand{\op}{\texttt{IntDiffOperations}}
\newcommand{\maple}{\textsc{Maple}}
\newcommand{\tma}{TH$\exists$OREM$\forall$}

\newcommand{\galg}{\mathcal{F}}

\newcommand{\cum}{{\textstyle \varint}}
\newcommand{\dops}{\galg\langle \partial \rangle}
\newcommand{\iops}{\galg\langle \cum \rangle}
\newcommand{\bops}{( \Phi )}
\newcommand{\idops}{\galg_{\Phi}\langle \partial, \cum \rangle}

\newcommand{\negskip}{~\vspace*{-3ex}}
\newcommand{\negsskip}{~\vspace*{-2ex}}

\begin{document}
\mainmatter
\title{Composing and Factoring Generalized Green's Operators and Ordinary Boundary Problems}
\author{Anja Korporal \and Georg Regensburger}
\authorrunning{A.~Korporal \and G.~Regensburger}
\titlerunning{Composing and Factoring Ordinary Boundary Problems}

\institute{Johann Radon Institute for Computational and Applied Mathematics (RICAM),\\
Austrian Academy of Sciences, Linz, Austria}

\maketitle

\begin{abstract}
 We consider solution operators of linear ordinary boundary problems with ``too many'' boundary conditions, which are not
 always solvable. These generalized Green's operators are a  certain kind of generalized
 inverses of differential operators. We answer the question when the product of two generalized Green's operators 
is again a generalized Green's operator for the product of the corresponding differential operators 
and which boundary problem it solves.
 Moreover, we show that---provided a
 factorization of the underlying differential operator---a generalized boundary problem
 can be factored into lower order problems corresponding to a factorization of the respective Green's operators.
 We illustrate our results by examples using the  \maple\ package \pack, where the presented algorithms  are implemented.
\end{abstract}


\section{Introduction}

Although linear boundary problems play an important role in
applied mathematics~\cite{Agarwal1986,Agarwal2008,Duffy2001,Stakgold1979},
there is little algebraic theory and algorithmic treatment of boundary problems.
Current computer algebra systems provide many symbolic tools for
differential equations, but boundary conditions are usually left to a backward solving procedure,
which---depending on the forcing function and on the conditions---may or may not work.

In~\cite{Rosenkranz2005}, a new operator based approach for symbolic computation with linear ordinary boundary problems was presented, 
which has constantly been extended over the last 
years~\cite{RegensburgerRosenkranz2009,RosenkranzRegensburger2008a,RosenkranzRegensburgerTecBuchberger2009}; 
see also~\cite{RosenkranzRegensburgerTecBuchberger2012} for a recent overview. 
The results needed are summarized in Section~\ref{sec:SC}.

The most recent algorithms for regular boundary problems (that are uniquely solvable) are implemented 
in the \tma\ system \cite{RosenkranzRegensburgerTecBuchberger2009,Tec2011},
and in the \maple\ package \pack~\cite{KorporalRegensburgerRosenkranz2011,KorporalRegensburgerRosenkranz2012,Korporal2012}.
They do not only allow to compute solution operators (Green's operators), but also to factor regular boundary problems 
into lower order problems provided a factorization of the underlying differential operator.
The factorization of boundary problems relies on the  multiplicative structure introduced in~\cite{RegensburgerRosenkranz2009},
which for regular problems corresponds to the multiplication of the respective solution operators in reverse order.

The \pack\ package also provides support for the class of singular problems treated in this paper:
We consider boundary problems where the differential equation per se is solvable, but where inconsistent 
boundary conditions allow solutions only for forcing functions satisfying suitable \emph{compatibility conditions}. As a simple example, consider the boundary problem 
 \begin{equation}
  \label{eq:bpsingular}
  \bvp{u''(x) = f(x)}{u(1)=u'(1)=u'(0)=0,}
 \end{equation}
where the forcing function $f$ clearly has to satisfy the compatibility condition  $\textstyle{\int_0^1} f(\xi)\dx\xi = 0$; see Example~\ref{ex:CC} for the corresponding code.

While Green's operators for regular boundary problems are right inverses of the differential operator, solution operators for singular problems can be described
as generalized inverses. 
Algorithms for computing such generalized Green's operators and the compatibility conditions of a singular boundary problem are 
presented in~\cite{KorporalRegensburgerRosenkranz2011}; we briefly recall the basic results in Section~\ref{sec:GGOP}. 
For singular boundary problems and 
generalized or modified Green's functions in ana\-ly\-sis, we refer for example to~\cite{Stakgold1979}
and~\cite{Loud1970}, and in the context of generalized inverses to
\cite[Sec.~9.4]{Ben-IsraelGreville2003}, \cite{Boichuk2004}, and~\cite[Sec.~H]{NashedRall1976}.

The goal of this paper is to extend the factorization algorithm for regular boundary 
problems~\cite{RegensburgerRosenkranz2009,RosenkranzRegensburgerTecBuchberger2009} to generalized
boundary problems. To this end, we have earlier investigated the multiplicative
structure of generalized inverses in~\cite{KorporalRegensburger2013}. It turns out that, in contrast to the regular case, the product of generalized Green's operators is not, in general, itself a 
generalized Green's operator. 

In Section~\ref{sec:Composition}, we define a new composition of generalized boundary problems, 
which includes the composition of regular problems, based on results from~\cite{Korporal2012,KorporalRegensburger2013}. 
Then we discuss algorithmic methods for testing when the so-called reverse order law holds, that is, when the product 
of  two generalized Green's operators is again a generalized Green's operator of the product of the corresponding differential operators. 
Moreover, if the reverse order law holds, we
can algorithmically determine which boundary problem is solved by the product of two generalized Green's operators.
We present a first implementation of the new algorithms in the framework of integro-differential operators  
and illustrate our results by examples, carried out using the \pack\ 
package\footnote{Available at \url{http://www.risc.jku.at/people/akorpora/index.html}.}.

Building on results of~\cite{Korporal2012}, we discuss in Section~\ref{sec:Factorization}
a new algorithm and implementation for factoring a generalized boundary problem, such that the 
factorization corresponds to the product of the respective generalized Green's operators. We illustrate the algorithm also with examples for differential equations with non-constant coefficients. The right-hand factor computed by this method will always be a regular boundary problem. However,
we also present some first steps for obtaining other possible factorizations in Section~\ref{sec:outlook}.

\section{Symbolic Computation for Boundary Problems} \label{sec:SC}

The algebraic framework for treating linear ordinary boundary problems with symbolic methods is given by the algebra of integro-differential
operators over an ordinary integro-differential algebra. This algebra was introduced  
in~\cite{Rosenkranz2005,RosenkranzRegensburger2008a} 
as a uniform language to express boundary 
problems---meaning differential equations and boundary conditions---as well as their Green's operators, 
which are integral operators.
We review the basic properties and refer the reader to~\cite{RosenkranzRegensburger2008a} 
and~\cite{RosenkranzRegensburgerTecBuchberger2012} for additional details. See also~\cite{Bavula2011,Bavula2013} for an extensive 
study on algebraic properties of integro-differential operators with polynomial coefficients and a single
evaluation (corresponding to initial value problems).

Extending \emph{differential algebras}~\cite{Kolchin1973}, where derivations are linear operators satisfying 
the Leibniz rule, \emph{integro-differential algebras} $(\galg, \partial, \cum)$ are defined as a differential 
algebra $(\galg, \partial)$ along with an ``integral'' $\cum$ 
that is a linear  right inverse of the derivation $\partial$ and satisfies an algebraic version of 
the \emph{integration by parts} formula. For the similar notion of differential Rota-Baxter algebras, see~\cite{GuoKeigher2008} and for a 
detailed comparison~\cite{GuoRegensburgerRosenkranz2013}. 

We call an integro-differential algebra over a field $K$ \emph{ordinary} if $\ker \partial = K$. 
The standard example of an ordinary integro-differential algebra is given by $\galg = C^\infty (\mathbb{R})$ with the usual
derivation and the integral operator $\cum \colon f \mapsto \cum_a^x f(\xi)\dx \xi$ for a fixed $a \in \R$.
 We call $\E = 1-\cum \circ \partial$ the \emph{evaluation} of
$\galg$. 
For representing not only initial value problems, but arbitrary boundary problems,
we include additional characters (multiplicative linear functionals)
$\E_c\colon f \mapsto f(c)$ at various evaluation
points $c \in \mathbb{R}$.

We write $\dops$ for the the ring of differential operators with coefficients in $\galg$ and $\iops$ for integral operators
of the form
$\sum_{i=1}^{n} f_i \cum g_i$ with $f_i, g_i \in \galg$.
For a set of characters $\Phi$, the corresponding (two-sided ideal of) \emph{boundary operators} $\bops$ are finite sums 
\begin{equation*}
   \sum_{\varphi \in \Phi} \left( \sum_{i \in \N}
      f_{i, \varphi} \, \varphi \partial^i + \sum_{j \in \N} g_{j, \varphi} \varphi \cum h_{j, \varphi} \right)
\end{equation*}
with $f_{i, \varphi}, g_{j, \varphi}, h_{j, \varphi} \in \galg$. 
 \emph{Stieltjes boundary conditions} are boundary operators where all the $f_{i, \varphi}\in K$ are constants and $ g_{j, \varphi}=1$, so that they act on $\galg$ as linear functionals.

The 
integro-differential operators $\idops$ are given as a direct sum of $K$-vector spaces
\begin{equation*}
 \idops = \dops \dotplus \iops \dotplus \bops.
\end{equation*}
The representation of integral operators and boundary conditions is not unique due to linearity of $\cum$, for 
normal forms of integro-differential operators; see~\cite{RosenkranzRegensburger2008a}.

For solving boundary problems, we restrict ourselves to monic (i.e., having leading coefficient $1$) differential 
operators  that have a \emph{regular fundamental system} $u_1, \ldots, u_n$, which means that the associated
Wronskian matrix
\begin{equation*}
  W(u_1,\ldots,u_n) = \begin{pmatrix}
    u_1 & \cdots & u_n\\
    u_1' & \cdots & u_n'\\
    \vdots & \ddots & \vdots\\
    u_1^{(n-1)} & \cdots & u_n^{(n-1)}
  \end{pmatrix}
\end{equation*}
is regular. Then the \emph{fundamental right inverse} (solving the initial value problem) can be computed
as an integro-differential operator in $\idops$ by the  
\emph{variation of constants} formula
\begin{equation} \label{eq:VarConst}
 \fri{T} = \sum_{i=1}^n u_i \cum d^{-1}\,d_i,
\end{equation}
 where $d$ is the determinant of the Wronskian matrix, $d_i= \det W_i$,
  and $W_i$ is the matrix obtained from $W$ by replacing the $i$th
  column by the $n$th unit vector.

Algorithms for solving and factoring
regular ordinary boundary problems are described in \cite{RosenkranzRegensburger2008a,RegensburgerRosenkranz2009}.
We recall the basic definitions and results.
A \emph{boundary problem} is defined as a pair $(T, \B)$ consisting of 
a monic differential operator  $T$ of order $n$ with a regular fundamental system and a space of boundary conditions
$\B = \Ll(\beta_1 , \ldots, \beta_n)$ generated by $n$ linearly independent Stieltjes boundary conditions. 
We think of $T$ as a surjective linear operator between suitable $K$-vector spaces of ``functions'' $T\colon V \to W$ and 
the boundary conditions $\beta_i$ acting as linear functionals from the dual space $V^*$.
We describe the subspace of functions satisfying the boundary conditions via the orthogonal
\begin{equation*}
 \Bo = \{f \in V \mid \beta(f) = 0 \text{ for all } \beta \in \B \},
\end{equation*}
i.e., $u \in V$ is a solution of $(T,  \B)$ for a given forcing function $f \in W$, if
\begin{equation*}
 Tu=f \quad \text{and} \quad u\in \Bo.
\end{equation*}
One easily checks that a boundary problem has a unique solution for each forcing function $f$ if and only if
$\Ker T \dotplus \Bo = V$: The sum $\Ker T + \Bo =V$ ensures the existence of a solution, the trivial intersection
$\Ker T \cap \Bo = \{0\}$ implies its uniqueness. Such boundary problems are called \emph{regular}, and 
the solution operator $G$ that maps each $f \in W$ to its unique solution $u \in V$ is called \emph{Green's operator}.
The Green's operator is a right inverse of $T$ and can be computed as $G = (1-P)\fri{T}$, 
where $P$ is the projector onto $\Ker T$ along $\Bo$, see also \cite{RosenkranzRegensburger2008a}.
We will use the notation $G = (T, \B)^{-1}$ for the Green's operator $G$ of a regular problem $(T, \B)$.

Regularity of a boundary problem can be tested via its \emph{evaluation matrix}: Let $(u_1, \ldots, u_n)$
be a basis of $\Ker T$ and $(\beta_1, \ldots, \beta_n)$ a basis of $\B$. Then $(T, \B)$ is regular iff
\begin{equation*}
 \beta(u) = \begin{pmatrix}
             \beta_1(u_1) & \ldots & \beta_1(u_n)\\
             \vdots & \ddots & \vdots\\
             \beta_n(u_1) & \ldots & \beta_n(u_n)
            \end{pmatrix} \in K^{n\times n}
\end{equation*}
is regular.

\newpage

For factoring regular boundary problems and their Green's operators, we recall the composition of boundary problems  in~\cite{RegensburgerRosenkranz2009}. This composition corresponds to the product of their Green's operators in reverse order. Since $G_1=(T_1, \B_1)^{-1}$ and $G_2=(T_2, \B_2)^{-1}$ are right inverses of $T_1 \colon V \to W$ 
and $T_2\colon U \to V$, the product
$G_2G_1$ obviously is a right inverse of $T_1T_2$. Defining the composition of boundary problems as
\begin{equation} \label{eq:compo}
 (T_1, \B_1) \circ (T_2, \B_2) = (T_1T_2, \B_2 + T_2^*(\B_1)),
\end{equation}
where $T_2^* \colon W^* \to V^*$ denotes the transpose map  $\beta \mapsto \beta \circ T_2$, the reverse order~law
\begin{equation} \label{eq:ROL}
 \bigl( (T_1, \B_1) \circ (T_2, \B_2) \bigr)^{-1}= (T_2, \B_2)^{-1} (T_1, \B_1)^{-1}
\end{equation}
always holds for regular problems.
Using this multiplicative structure of boundary problems, it is always possible to split a regular boundary problem $(T, \B)$
into regular lower order problems, provided there exists a factorization $T= T_1T_2$ of
the differential operator; see~\cite{RosenkranzRegensburgerTecBuchberger2009}
for a constructive proof that requires only a fundamental system of $T_2$.

We conclude this section with a remark on the ``function spaces'' $V$ and $W$ on which we let the differential
operator $T$ act. The assumption $V=W=\galg$ used for example 
in~\cite{RosenkranzRegensburger2008a,KorporalRegensburgerRosenkranz2011}---ensuring well-definedness of arbitrary 
operations---for some applications is too restrictive. For a given boundary problem $(T, \B)$ of order $n$, 
it is for example sufficient to consider $n$ times continuously differentiable functions, i.e., $V = C^n[a,b]$ 
and $W =C[a,b]$, where $a$ and $b$ are the minimal and maximal evaluation point appearing in the Stieltjes conditions 
of $\B$. 
Similarly, for composing boundary problems $(T_1, \B_1)$ and $(T_2, \B_2)$ of order $m$ and $n$, it suffices to
restrict the domains as to consider $T_2 \colon C^{m+n}[a,b] \to C^m[a,b]$ and $T_1 \colon C^m[a,b] \to C[a,b]$ 
(with suitable choices of $a$ and $b$).

\section{Generalized Green's Operators} \label{sec:GGOP}

In~\cite{KorporalRegensburgerRosenkranz2011}, several methods of the previous section are generalized
to boundary problems with ``too many'' boundary conditions, meaning that $\ord T < \dim \B$.
These problems are not solvable for all forcing functions $f$; but we keep the condition $\Ker T \cap \Bo = \{0\}$
to ensure unique solutions. We briefly recall the basic results from~\cite{KorporalRegensburgerRosenkranz2011}.

\begin{definition}
 We call a boundary problem $(T, \B)$ \emph{semi-regular} if 
 \begin{equation} 
  \Ker T \cap \Bo = \{0\}.
 \end{equation}
 Let $(T, \B)$ be a boundary problem with $T\colon V \to W$ and $E \leq W$. We call the triple $(T, \B, E)$ regular, if
 $(T, \B)$ is  semi-regular and 
 \begin{equation} \label{eq:DirSum2}
  T(\Bo) \dotplus E = W.
 \end{equation}
 In this case, we call $E$ an \emph{exceptional space} for $(T, \B)$ and $(T, \B, E)$ a \emph{generalized 
 boundary problem}.
\end{definition}

For simplicity, we also refer to the triple $(T,\B,E)$ as a boundary problem and identify
$(T,\B)$ with $(T,\B,\{0\})$. One easily checks that regularity of $(T, \B)$ implies regularity of
$(T, \B, \{0\})$ and vice versa, see also~\cite[Sec.~4.1]{Korporal2012}.

For solving generalized boundary problems in analysis 
the forcing function $f$ is projected onto $T(\Bo)$, the space of \emph{admissible forcing functions};  see for example~\cite{Stakgold1979,Loud1970}. This leads to the following algebraic definition of generalized Green's operators.

\begin{definition} 
 Let $(T, \B, E)$ be regular with $T \colon V \to W$, and let $Q$ be the projector onto $T(\Bo)$ along $E$. 
 Then $u \in V$ is called a \emph{solution} for $f \in W$ if
 \begin{equation} \label{eq:GenSolution}
  Tu =Qf \quad \text{and} \quad u \in \Bo.
 \end{equation}
 The \emph{generalized Green's operator} maps each $f \in W$ to the unique solution 
 according to \eqref{eq:GenSolution}.
\end{definition}

As in the regular case, we use the notation $G = (T, \B, E)^{-1}$ for the generalized Green's operator for a regular boundary problem $(T, \B,E)$.

For testing semi-regularity of a boundary problem $(T, \B)$ with $\ord T =m$ and $\dim \B =n$, 
we have to check whether the associated evaluation matrix 
$\beta(u) \in K^{n\times m}$ has full column rank, see also~\cite[Lem.~1]{KorporalRegensburgerRosenkranz2011}.
Condition~\eqref{eq:DirSum2} can be checked analogously, provided an implicit description of $T(\Bo)$ via 
\emph{compatibility conditions} $\C = T(\Bo)^{\perp}$:
If $(\gamma_1, \ldots, \gamma_r)$ is a basis of $\C$ and $(e_1, \ldots, e_r)$ is a basis of $E$, then the
corresponding evaluation matrix $\gamma(e) \in K^{r\times r}$ has to be regular.
Moreover, in~\cite[Prop.~1]{KorporalRegensburgerRosenkranz2011} a method is presented to compute a 
basis of $\C$; we have
\begin{equation}\label{eq:CompCond}
 \C = G^*(\B \cap (\Ker T)^{\perp})
\end{equation}
for any right inverse $G$ of $T$. 
Note that equation~\eqref{eq:CompCond} requires computing the intersection of the finite dimensional space $\B$
with the finite codimensional space $(\Ker T)^{\perp}$, which can be done using the following observation; see for example~\cite{KorporalRegensburger2013}.

\begin{lemma} \label{lem:intersections}
 Let the subspaces $U \leq V$ and $\B \leq V^*$ be generated respectively by $u=(u_1, \ldots, u_m)$ 
 and $\beta=(\beta_1, \ldots, \beta_n)$.
 Let $k^1, \ldots, k^r \in F^m$ be a basis of $\Ker \beta(u)$, and $\kappa^1, \ldots, \kappa^s \in F^n$ a basis 
 of $\Ker (\beta(u))^T$.   Then 
 \begin{enumerate}[label=(\roman*)]
\item $U\cap \Bo$ is generated by
   $\sum_{i=1}^m k^1_i u_i , \ldots, \sum_{i=1}^m k^r_i u_i$ and
 \item $U^{\perp} \cap \B $ is generated by
  $\sum_{i=1}^n \kappa^1_i\beta_i , \ldots, \sum_{i=1}^n \kappa^s_i \beta_i$.
 \end{enumerate}
\end{lemma}

\begin{example} \label{ex:CC}
 We consider the boundary problem \eqref{eq:bpsingular}
 \begin{equation*}
  \bvp{u''(x) = f(x)}{u(1)=u'(1)=u'(0)=0,}
 \end{equation*}
which reads in operator notation 
as $(\partial^2, \Ll(\E_1, \E_1\partial, \E_0\partial))$. We employ the 
standard integro-differential algebra $C^{\infty}(\R)$ with the 
  usual derivation and integral operator $\cum \colon f \mapsto \cum_0^x f(\xi)\dx\xi$ that
  is implemented in the \maple\ package \pack\ as presented 
  in~\cite{KorporalRegensburgerRosenkranz2011,KorporalRegensburgerRosenkranz2012}.

 In \cite{KorporalRegensburgerRosenkranz2011}, we have already computed the compatibility conditions and
 the generalized Green's operator with respect to the exceptional space $E=\R$. 
 We again carry out the computations in the \pack\ package, but using 
the interface \op\ for input of integro-differential operators. There, we use the symbols $d$ and $a$ for input of the 
differential and integral operator, and $e(c)$ for the evaluation at $c \in \R$. 
For the \maple\ output, the respective capital letters $\D$, $\A$, and $\E[c]$ are used,
and the non-commutative multiplication of
integro-differential operators is denoted by ``$.$''. The constructors $\BP, \GBP, \BC$ and $\ES$ are used for input
of respectively boundary problems, generalized boundary problems, boundary conditions, and exceptional spaces.
 \begin{center}
\begin{minipage}{0.999\textwidth}
\begin{framed}~\\[-1cm]\small
\begin{alltt}
> with(IntDiffOp):
> with(IntDiffOperations):
> t1 := d^2:    #input differential operator
> b1 := e(1): b2 := e(1).d: b3 := e(0).d:    #boundary conditions	
> B1 := BC(b1, b2, b3):
> C1 := CompatibilityConditions(BP(t1, B1)); \negsskip
 \[\BC(\E[1].\A)\] \negskip
> E1 := ES(1):    #exceptional space
> bp1 := GBP(t1, B1, E1):    #generalized boundary problem
> g1 := GreensOperator(bp1);
\end{alltt} 
\begin{equation*}
 x.\A - \A.x + \bigl(-\frac{1}{2}x^2-\frac{1}{2}\bigr).\E_1. \A + \E_1 . \A. x
\end{equation*} \negsskip
\end{framed}\normalsize
\end{minipage}
\end{center}

\end{example}

\section{Composition of Generalized Green's Operators} \label{sec:Composition}

In Section~\ref{sec:SC}, we have already recalled the multiplicative structure for boundary problems
introduced 
in~\cite{RosenkranzRegensburger2008a,RegensburgerRosenkranz2009}. 
In contrast to the regular case, where Green's operators
always satisfy the reverse order law, the situation is more involved for generalized Green's operators
since they are not right inverses of the differential operator.

\begin{proposition}
 \label{prop:GOI}
 Let $(T,\B, E)$ be regular with $T \colon V \to W$ and $G = (T, \B, E)^{-1}$ its generalized Green's operator. 
 Then $GTG=G$, that is, $G$ is an outer inverse of $T$.  
\end{proposition}

\begin{proof}
  By definition of generalized Green's operators, we  have $Tu=Qf$ for all $f \in W$, 
  as well as $Gf=GQf = u$, where $Q$ denotes the projector
 onto $T(\Bo)$ along $E$. Hence $TGf = Tu = Qf$, and $GTGf = GQf = Gf$ for all $f \in W$.
\end{proof}

In terms of generalized inverses, the Green's operator of a regular problem $(T, \B, E)$ can therefore also be defined as the unique outer inverse $G$ of $T$ with $\img G = \Bo$ and $\Ker G = E$.
In particular, for an outer inverse $G$ of $T$, the boundary problem $(T, (\img G)^{\perp}, \Ker G)$  
is regular, and $G$ is its generalized Green's operator; see also \cite[Rmk.~4.7]{Korporal2012}.

The composition $G_2G_1$ of two outer inverses of $T_1$ and $T_2$ is in general not an
outer inverse of the product $T_1T_2$. However, from the above considerations it is clear that if 
$G_2G_1$ is an outer inverse of $T_1T_2$,
then computing its kernel and image yields the boundary problem it solves. For a proof of the 
following result, see \cite[Thm.~6.2]{KorporalRegensburger2013} or~\cite[Thm.~3.27]{Korporal2012}.

\begin{theorem} \label{prop:ROL}
 Let $(T_1, \B_1, E_1)$ and $(T_2, \B_2, E_2)$ be regular with  $T_1 \colon V \to W$, $T_2 \colon U\to V$ 
 and $G_1 = (T_1, \B_1, E_1)^{-1}$, $G_2 = (T_2, \B_2, E_2)^{-1}$ their generalized Green's operators.
 If $G_2G_1$ is an outer inverse of $T_1T_2$, the boundary problem
 \begin{equation}\label{eq:Comp1}
  \bigl(T_1T_2, \B_2 + T_2^*(\B_1 \cap E_2^{\perp}), E_1  + T_1(\Bo_1 \cap E_2) \bigr)
 \end{equation}
 is regular with generalized Green's operator $G_2G_1$. 
 Furthermore, the two sums $\B_2 \dotplus T_2^*(\B_1 \cap E_2^{\perp})$ and $E_1  \dotplus T_1(\Bo_1 \cap E_2)$
 are direct.
\end{theorem}

Based on Equation~\eqref{eq:Comp1}, we define the composition of two arbitrary boundary problems as follows. 
\begin{definition}
The \emph{composition} of two boundary problems is defined as
\begin{equation}
\label{eq:compogen}
 (T_1, \B_1, E_1) \circ (T_2,  \B_2, E_2) \\
 = \bigl(T_1T_2, \B_2 + T_2^*(\B_1 \cap E_2^{\perp}), E_1  + T_1(\Bo_1 \cap E_2) \bigr),
\end{equation}
assuming that the composition $T_1T_2$ is defined. 
\end{definition}
This definition clearly reduces to the composition of regular boundary problems~\eqref{eq:compo} when $E_1=E_2=\{0\}$.  
The composition of generalized boundary problems is implemented in the following algorithm.

\begin{framed}~\\[-1cm]
\begin{alg}[Composition]~ \label{alg:Comp}
 \begin{description}
 \item[Input] Two  boundary problems $(T_1, \B_1, E_1)$ and $(T_2, \B_2, E_2)$, \\
 $\beta_1, \ldots, \beta_n$ and $\tilde{\beta}_1, \ldots, \tilde{\beta_{\nu}}$ bases of $\B_1$ and $\B_2$, \\ 
 $e_1, \ldots, e_t$ and $\tilde{e}_1, \ldots,  \tilde{e}_{\tau}$ bases of $E_1$ and $E_2$. 
 \item[Output] The composite boundary problem $(T_1, \B_1, E_1)\circ (T_2, \B_2, E_2)$.
\end{description}
\begin{enumerate}
 \item Multiply $T=T_1T_2 \in \galg_{\Phi}\langle \partial, \cum \rangle$.
 \item Compute a basis $\gamma_1, \ldots, \gamma_k$ of $\B_1 \cap E_2^{\perp}$ using Lemma~\ref{lem:intersections}.
 \item Compute a basis $v_1, \ldots, v_{\ell}$ of $\Bo_1 \cap E_2$ using Lemma~\ref{lem:intersections}.
 \item For $1\leq i \leq k$ multiply $\delta_i = \gamma_i T_2 \in \galg_{\Phi}\langle \partial, \cum \rangle$.
 \item For $1 \leq j \leq \ell$ compute $t_j = T_1(v_j) $. % \Ll(\gamma_1, \ldots, \gamma_r) + K$.
 \item Compute a basis $\alpha_1, \ldots, \alpha_q$ of $\Ll(\tilde{\beta}_1, \ldots, \tilde{\beta_\nu}, \delta_1, \ldots, \delta_k)$.
 \item Compute a basis $f_1, \ldots, f_r$ of $\Ll(e_1, \ldots, e_t, t_1, \ldots, t_{\ell})$.
 \item Return $(T, (\alpha_1, \ldots, \alpha_q), (f_1, \ldots, f_r))$.
\end{enumerate}
\end{alg}~\\[-1.2cm]
\end{framed}

We are especially interested in the situation of Theorem~\ref{prop:ROL}, that is, the case where the 
composition of boundary problems corresponds to the composition of their generalized Green's operators.
For testing when $G_2G_1$ is an outer inverse of $T_1T_2$, we use the following characterization from~\cite{KorporalRegensburger2013,Korporal2012}, which is based on  results from~\cite{GrossTrenkler1998} and \cite{Werner1992}. It gives necessary and sufficient 
conditions on the subspaces $\Bo_1$, $T_2(\Bo_2)$, $E_2$,  and $T_1^{-1}(E_1)$  
such that the revers order law  
 \begin{equation}
  \label{eq:ROLgen}
   ((T_1, \B_1, E_1)\circ (T_2, \B_2, E_2))^{-1}=(T_2, \B_2, E_2)^{-1} (T_1, \B_1, E_1)^{-1}
 \end{equation}
for the respective generalized Green's operators holds.  
The conditions can be checked using only the input data, in particular, without computing $G_2$ or $G_1$.

\begin{theorem} \label{thm:G2G1OI}
 Let $(T_1, \B_1, E_1)$ and $(T_2, \B_2, E_2)$ be regular with  $T_1 \colon V \to W$, $T_2 \colon U\to V$ 
 and $G_1 = (T_1, \B_1, E_1)^{-1}$, $G_2 = (T_2, \B_2, E_2)^{-1}$ their generalized Green's operators. Let 
 $\C_2 = T_2(\Bo_2)^{\perp}$ and $K_1 = T_1^{-1}(E_1)$.
The following conditions are equivalent:
\begin{enumerate} [label=(\roman*)]
 \item $G_2G_1$ is an outer inverse of $T_1 T_2$,
 \item $(T_1, \B_1, E_1) \circ (T_2, \B_2, E_2)$ is regular with Green's operator $G_2G_1$,		 	\label{algOI1}	
 \item $\C_2 + (\B_1 \cap E_2^{\perp}) \geq \B_1 \cap (E_2 \cap K_1)^{\perp}$, \label{algOI2}	
 \item $\B_1 \geq \C_2 \cap (E_2 \cap \Bo_1)^{\perp} \cap (E_2 \cap  K_1)^{\perp}$, \label{algOI3}	
 \item $ K_1  \dotplus (E_2 \cap \Bo_1) \geq E_2 \cap (\B_1 \cap \C_2)^{\perp}$, \label{algOI4}	
 \item $E_2 \geq  K_1 \cap (\B_1 \cap E_2^{\perp})^{\perp} \cap (\B_1 \cap \C_2)^{\perp}$. \label{algOI5}	
\end{enumerate}
\end{theorem}

All conditions are formulated so that they can be checked algorithmically with the methods from the end of 
Section~\ref{sec:SC}.
For non-algorithmic purposes, also the equivalent dual conditions are useful, as for example
\begin{equation} \label{eq:OI3}
 \Bo_1 \leq \Co_2 \dotplus (E_2 \cap \Bo_1) \dotplus (E_2 \cap K_1),
 \end{equation}
which is the dual statement of Condition~\ref{algOI3}.

For testing the conditions of Theorem \ref{thm:G2G1OI}, it is necessary to determine 
the space $K_1=T_1^{-1}(E_1)$. This can be done
using the identity 
\begin{equation} \label{eq:EasyKerComp} 
 T^{-1}(E) = G(E) \dotplus \Ker T
 \end{equation}
for the inverse image of a subspace, where $G$ is an arbitrary right inverse  of $T$. It can be verified directly, 
see also~\cite[Prop.~A.12]{RegensburgerRosenkranz2009}.  As always, we assume a
regular fundamental system to be given, which allows to construct the fundamental right inverse. Since any right inverse of $T$ is injective and because of the direct sum in \eqref{eq:EasyKerComp}, the output of the following algorithm is indeed a basis of $T^{-1}(E)$.

\begin{framed}~\\[-1cm]
\begin{alg}[Inverse Image]~ \label{alg:preimage}
 \begin{description}
 \item[Input] A monic differential operator $T\colon V \to W$ and a basis $e_1, \ldots, e_r$ of $E \leq W$.
 A regular fundamental system $s_1, \ldots, s_m$ of $T$.
 \item[Output] A basis of $T^{-1}(E)$.
\end{description}
\begin{enumerate}
 \item Compute the fundamental right inverse $\fri{T}$ according to \eqref{eq:VarConst}.
 \item For $1 \leq i \leq r$ compute $k_i = \fri{T} (e_i)$. 
 \item Return $s_1, \ldots, s_m, k_1, \ldots, k_r$.
\end{enumerate}
\end{alg}~\\[-1.2cm]
\end{framed}

For testing the necessary and sufficient conditions of Theorem \ref{thm:G2G1OI}, we assume that for 
finite-dimensional subspaces of $V$ or $V^*$, we can compute sums and intersections and 
check inclusions. For the heuristic approach used in the \pack\ package, see~\cite[Sec.~7.4]{Korporal2012}. 

The following algorithm implements a test of Condition~\ref{algOI2}, the others can be
implemented similarly. (None of them seems to be particularly preferable from computational aspects.)

\begin{framed}~\\[-1cm]
\begin{alg}[Check Reverse Order Law]~ \label{alg:checkROL}
 \begin{description}
 \item[Input] Two  boundary problems $(T_1, \B_1, E_1)$ and $(T_2, \B_2, E_2)$, \\
 $\beta_1, \ldots, \beta_n$ and $\tilde{\beta}_1, \ldots, \tilde{\beta_{\nu}}$ bases of $\B_1$ and $\B_2$, \\ 
 $e_1, \ldots, e_t$ and $\tilde{e}_1, \ldots,  \tilde{e}_{\tau}$ bases of $E_1$ and $E_2$. \\
 Regular fundamental systems $s_1, \ldots, s_m$ of $T_1$ and $\tilde{s}_1, \ldots, \tilde{s}_{\ell}$ of $T_2$.
 \item[Output] \textbf{true} if the reverse order law holds and \textbf{false} otherwise.
\end{description}
\begin{enumerate}
 \item Compute a basis of $T_1^{-1}(E_1)$ with Algorithm~\ref{alg:preimage}.
 \item Compute a basis of $J = E_2 \cap T_1^{-1}(E_1)$.
 \item Compute a basis of $B = \B_1 \cap J^{\perp}$ using Lemma~\ref{lem:intersections}.
 \item Compute a basis of $K = \B_1 \cap E_2^{\perp}$ using Lemma~\ref{lem:intersections}.
 \item Compute the compatibility conditions $\gamma_1, \ldots, \gamma_r$ of $(T_2, \B_2)$ due to \eqref{eq:CompCond}.
 \item Compute $C = \Ll(\gamma_1, \ldots, \gamma_r) + K$.
 \item If $B\leq C$ return \textbf{true}, else return \textbf{false}.
\end{enumerate}
\end{alg}~\\[-1.2cm]
\end{framed}

\begin{example} \label{ex:GIpaper}
As a first example, we consider the generalized boundary problem 
 $(\D^2, \Ll(\E_1, \E_1\D, \E_0\D), \R)$
 from Example \ref{ex:CC}. As a second boundary problem, we consider the differential operator 
$T_2 = \D^2 - 1$, also with the boundary conditions $\Ll(\E_1, \E_1\D, \E_0\D)$, or, in usual notation
\begin{equation*}
 \bvp{u''(x) - u(x) = f(x)}{u(1)=u'(1)=u'(0)=0.}
\end{equation*}
As we will see below, the forcing function $f$ has to satisfy the  compatibility condition 
\begin{equation*}
 \textstyle{\int_0^1} \exp(-\xi)f(\xi)\dx\xi + \textstyle{\int_0^1} \exp(\xi)f(\xi)\dx\xi =0,
\end{equation*}
 so that $E_2= \Ll(x)$ is
a suitable exceptional space for the second boundary problem.
Moreover, we compute the inverse image $K_1=T_1^{-1}(E_1)$ for testing the 
conditions of Theorem~\ref{thm:G2G1OI}.
\begin{center}
\begin{minipage}{0.999\textwidth}
\begin{framed}~\\[-1cm]\small
\begin{alltt}
> t2 := d^2-1: B2 := B1:
> C2 := CompatibilityConditions(BP(t2, B2))); \negsskip
 \[ \BC(\E[1].\A.\exp(-x)+\E[1].\A.\exp(x))  \] \negskip
> E2 := ES(x): bp2 := GBP(t2, B2, E2):
> IsRegular(GBP(t2, B2, E2)); \negskip
\[true\] \negskip 
> K1 := InverseImage(t1, E1);
\end{alltt} \negskip
\begin{equation*} 
 \ES(1, x, x^2)
\end{equation*} \negskip
\end{framed}\normalsize
\end{minipage}
\end{center}

In this case $E_2 =\Ll(x) \leq \Ll(1,x,x^2) =K_1$, so Condition (iv) of Theorem~\ref{thm:G2G1OI}
is trivially satisfied. We check that the reverse order law holds for the respective Green's operators.

\begin{center}
\begin{minipage}{0.999\textwidth}
\begin{framed}~\\[-1cm]\small
\begin{alltt} \negsskip
 > p1 := MultiplyBoundaryProblem(bp1, bp2);
\end{alltt} \scriptsize
\begin{equation*}
 \GBP(\D^4-\D^2, \BC(\E[0].\D, \E[0].\D^3-\E[1].\D^3, \E[1], \E[1].\D, \E[1].\D^2-\E[1].\D^3), \ES(1))
\end{equation*} \small
\begin{alltt}
 > g := GreensOperator(p1):
 > g2 := GreensOperator(bp2):
 > IsZero(g-g2.g1);    #check reverse order law \negskip
 \[true\] \negskip
 \end{alltt} \negskip \negskip \negskip
\end{framed}\normalsize
\end{minipage}
\end{center}

Since for the boundary problems in the previous example, we have $T_1T_2 = T_2T_1 = \D^4 -\D^2$, we can also
consider the product of Green's operators in reverse order, that is, test if $(T_1, \B_1, E_1)^{-1}(T_2, \B_2, E_2)^{-1}$
is an outer inverse of $T_1T_2$. 

\begin{example} \label{ex:GI2}
We follow the steps of Algorithm \ref{alg:checkROL}, interchanging the indices accordingly. 
Recall from Example~\ref{ex:CC} that $\C_1$ is generated by $\E_1\A$ and
that we have chosen the exceptional space $E_1 = \R$. 

 \begin{center}
\begin{minipage}{0.999\textwidth}
\begin{framed}~\\[-1cm]\small
\begin{alltt}
 > K2 := InverseImage(t2, E2);
\end{alltt}
\begin{equation*} 
 \ES(x, \exp(x), \exp(-x))
\end{equation*}
\begin{alltt}
 > J := Intersection(E1, K2); \negsskip
 \[ \ES() \] \negsskip
 > B := Intersection(B2, J, space = dual);
 \[ \BC(\E[1], \E[1].\D, \E[0].\D)  \]\negsskip
 > K := Intersection(B2, E1, space = dual);
 \[ \BC(\E[1].\D, \E[0].\D)  \]\negsskip
 > TestComposition(bp2, bp1);    #check reverse order law \negsskip
 \[false\]\negskip \negskip 
\end{alltt}
\end{framed}\normalsize
\end{minipage}
\end{center}
\end{example}
We see that in the above example Algorithm \ref{alg:checkROL} returns \emph{false} since
the inclusion 
\begin{equation*} 
 \C_1 + (\B_2 \cap E_1^{\perp}) \geq \B_2 \cap (E_1 \cap K_2)^{\perp}
\end{equation*}
from Theorem \ref{thm:G2G1OI} (ii) is not satisfied: the left-hand side does not contain
the evaluation $\E_1$.
We also verify directly that the product $G_1G_2$ indeed is not an outer inverse of $T_1T_2$.

\begin{center}
\begin{minipage}{0.999\textwidth}
\begin{framed}~\\[-1cm]\small
\begin{alltt} \negsskip
 > t := t2.t1: g:= g1.g2:
 > IsZero((g.t).g-g); \negskip
  \[false\] \negskip
 > p2 := MultiplyBoundaryProblem(bp2, bp1);
\end{alltt} \scriptsize
\begin{equation*}
 \GBP(\D^4-\D^2, \BC(\E[0].\D, \E[0].\D^3, \E[1], \E[1].\D, \E[1].\D^3), \ES(x))
\end{equation*} \small
 \begin{alltt}
 > IsRegular(p2); \negskip
 \[true\]\negskip \negskip
\end{alltt} 
\end{framed}\normalsize
\end{minipage}
\end{center}

Although $G_1G_2$ is not an outer inverse of $T_2T_1$, in this case the composite boundary problem
$(T_2, \B_2, E_2) \circ (T_1, \B_1, E_1)$ is regular. In general, we do not even obtain semi-regularity
of $(T_1T_2, \B_2 + T_2^*(\B_1 \cap E_2^{\perp}))$, see~\cite[Thm.~4.21]{Korporal2012} for more details.
\end{example}

\section{Factorization of Generalized Boundary Problems} \label{sec:Factorization}

In this section, we discuss certain classes of factorizations of generalized boundary problems.
We start with a regular boundary problem $(T, \B, E)$ and a factorization $T= T_1T_2$ of the underlying differential operator.  
The factorization algorithm presented in this section works for arbitrary factorizations of $T$, regardless of the 
coefficient domain, as long as a regular fundamental system of $T_2$ is known.
In our package, we use the function \texttt{DFactor} from the \maple\ \texttt{DETools} package to factor differential
 operators with rational coefficients (Example \ref{ex:RatCoeff}), unless a particular factorization is specified in the input.

As in \cite{RegensburgerRosenkranz2009} for regular boundary problems, the overall goal would be to characterize all regular boundary problems $(T_1, \B_1, E_1)$ and $(T_2, \B_2, E_2)$ such that
\begin{equation*}
 (T, \B, E) = (T_1, \B_1, E_1) \circ (T_2, \B_2, E_2).
\end{equation*}
For generalized boundary problems, we additionally have to require that the reverse order law
\begin{equation*}
 ((T_1, \B_1, E_1) \circ (T_2, \B_2, E_2))^{-1} = (T_2, \B_2, E_2)^{-1} (T_1, \B_1, E_1)^{-1}
\end{equation*}
holds, which is always valid for regular boundary problems~\eqref{eq:ROL}. Due to the structure of  the composite~\eqref{eq:compogen}, where information about $\B_1$ gets lost when intersecting
with $E_2$, and the rather involved interactions of the subspaces $\Bo_1, K_1, \Co_2$ and $E_2$ in Theorem~\ref{thm:G2G1OI}, it is a difficult task to describe \emph{all} possible factorizations. 

In the following, we discuss a special case of factoring a regular boundary problem $(T, \B, E)$:  
As a right factor, we strive for a regular problem $(T_2, \B_2)$, meaning that $E_2=\{0\}$. Then 
the reverse order law is trivially satisfied, 
and the composite takes the easier form
\begin{equation*}
 (T_1, \B_1, E_1) \circ (T_2, \B_2) = (T_1T_2, \B_2 \dotplus T_2^*(\B_1), E_1).
\end{equation*}
The existence of such a factorization was proven in~\cite{Korporal2012}. It relies on the fact that
for a semi-regular boundary problem $(T, \B)$ and a factorization $T=T_1T_2$, there always exists
$\B_2 \leq \B$ such that $(T_2, \B)$ is regular. 
(For ordinary differential equations, this can be seen immediately by inspecting the evaluation matrix.)

\begin{theorem} \label{thm:FactReg}
 Let $(T, \B)$ be a semi-regular boundary problem and let $T= T_1T_2$ be a factorization of $T$ into 
 monic differential operators. 
   For each exceptional space $E$ of $(T, \B)$ there exists a unique regular boundary problem
 $(T_1, \B_1, E)$ such that for each $\B_2 \leq \B$ for which $(T_2, \B_2)$ is regular, we have
 \begin{equation} \label{eq:RegLeft}
  (T, \B, E) = (T_1, \B_1, E) \circ (T_2, \B_2) = (T_1T_2, \B_2 \dotplus T_2^*(\B_1), E)
 \end{equation}
 and $(T, \B, E)^{-1} =  (T_2, \B_2)^{-1} (T_1, \B_1, E)^{-1}$. The boundary conditions of the left 
 factor are given by $\B_1 = H_2^*(\B \cap (\Ker T_2)^{\perp})$, where $H_2$ is
 an arbitrary right inverse of $T_2$.

\end{theorem}

\begin{example} \label{ex:fac1}
We consider the factorization of the boundary problem \texttt{p2} from Example~\ref{ex:GI2} with
$\D^4-\D^2 = (\D^2 -1)\cdot \D^2$:
\begin{multline*}
 (\D^4-\D^2, \Ll(\E_0\D, \E_0\D^3, \E_1, \E_1\D, \E_1\D^3), \Ll(x)) \\
 = (\D^2 -1, \Ll(\E_0\D, \E_1 \D, \E_1\A), \Ll(x)) \circ
 (\D^2, \Ll(\E_0\D, \E_1)).
\end{multline*}
The following \maple\ session shows that the reverse order law holds for the respective Green's operators.

\begin{center}
\begin{minipage}{0.999\textwidth}
\begin{framed}~\\[-1cm]\small
\begin{alltt} \negsskip
 > p2; \negsskip
\end{alltt}  \scriptsize 
\begin{equation*}
 \GBP(\D^4-\D^2, \BC(\E[0].\D, \E[0].\D^3, \E[1], \E[1].\D, \E[1].\D^3), \ES(x))
\end{equation*} \small
\begin{alltt}
 > g := GreensOperator(p2):
 > f1, f2 := FactorBoundaryProblem(p2, t2, t1); \negskip
 \end{alltt}\scriptsize
 \begin{equation*}
  \GBP(D^2-1, \BC(\E[0].\D, \E[1].\D, \E[1].\A), \ES(x)), \BP(\D^2, \BC(\E[0].\D, \E[1]))
 \end{equation*}\small
\begin{alltt}
 > g1 := GreensOperator(f1): g2:= GreensOperator(f2):
 > IsZero(g-g2.g1); \negskip
 \[true\]
\end{alltt} \negskip \negskip \negskip
\end{framed}\normalsize
\end{minipage}
\end{center}
\end{example}

We now show how to construct the above factorization algorithmically, 
generalizing the method presented in  \cite{RosenkranzRegensburgerTecBuchberger2009}.
We assume a given factorization $T= T_1T_2$ of the differential operator and a regular fundamental system of
$T_2$, which is needed in Step 4 to construct a right inverse of $T_2$.

\begin{framed}~\\[-1cm]
\begin{alg}[Right Regular Factorization]~ \label{alg:RightFact}
\begin{description}
 \item[Input] A regular boundary problem $(T, \B, E)$, \\
 $\beta_1, \ldots, \beta_{n}$ a basis of $\B$,  
 $e_1, \ldots, e_r$ a basis of $E$.  \\ A factorization $T = T_1T_2$,  a reg.~fundamental system $s_1, \ldots, s_{\mu}$ 
 of $T_2$. 
 \item[Output] Two regular boundary problems $(T_1, \B_1, E)$ and $(T_2, \B_2)$ with $(T, \B, E) = (T_1, \B_1, E) \circ (T_2, \B_2)$.
\end{description}
\begin{enumerate}
 \item Compute the evaluation matrix $M = \beta(s) \in F^{n \times \mu}$.
 \item Compute $S = (s_{i,j}) \in F^{n \times n}$ s.t. ~$SM$ is in reduced row echelon form. \label{step:LUDec}
 \item For $1 \leq i \leq n$ set $\tilde{\beta}_i = \sum_{k=1}^n s_{i,k} \beta_k$.
 \item Compute a right inverse $H_2$ of $T_2$ according to \eqref{eq:VarConst}.
 \item For $\mu+1 \leq j \leq n$ multiply $\alpha_{j-\mu} = \tilde{\beta}_j H_2 \in \galg_{\Phi}\langle \partial, \cum \rangle$. \label{step:Hright}
 \item Return $(T_1, (\alpha_1, \ldots, \alpha_{n-\mu}), (e_1, \ldots, e_r))$ and 
 $(T_2, (\tilde{\beta}_1, \ldots, \tilde{\beta}_\mu))$.
\end{enumerate}
\end{alg}~\\[-1.2cm]
\end{framed}

\begin{theorem}
 Let $(T, \B, E)$ be regular and $T = T_1T_2$.
 The boundary problems $(T_1, \B_1, E)$ and $(T_2, \B_2)$ computed by Algorithm~\ref{alg:RightFact} are 
 regular and satisfy 
 \begin{equation*} 
  (T, \B, E) = (T_1, \B_1, E) \circ (T_2, \B_2).
 \end{equation*}
\end{theorem}

\begin{proof}
 In view of Theorem~\ref{thm:FactReg}, we only have to show that $(T_2, \Ll(\tilde{\beta}_1, \ldots , \tilde{\beta}_{\mu}))$ 
 is regular with $ \tilde{\beta}_i \in \B$ for $1\leq i \leq \mu$,  and that the Stieltjes conditions
 $\tilde{\beta}_{\mu+1}, \ldots, \tilde{\beta}_n$ computed in Step 3, are a basis of $\B \cap (\Ker T_2)^{\perp}$.
 
 First we observe that in Step 3 we obviously have $\tilde{\beta}_i \in \B$ for all $i$. Since $\Ker T_2 \leq \Ker T$ and since $(T, \B)$ is semi-regular, $(T_2, \B)$ is also semi-regular.
 Hence the evaluation matrix $M$ computed in Step 1 has rank $\mu$,  
 and the first $\mu$ rows of $SM$ in Step 2 give the $\mu \times \mu$ identity matrix. 
 Since $SM$ is the evaluation matrix $\tilde{\beta}(s)$---meaning that the $\tilde{\beta}$ are applied to the fundamental system of 
 $\Ker T_2$---the boundary problem $(T_2, \Ll(\tilde{\beta}_1, \ldots , \tilde{\beta}_{\mu}))$ is regular.
 
 The lower part of $SM$ is the $(n-\mu) \times \mu$ zero matrix, hence we know that
 \begin{equation*} 
  \Ll(\tilde{\beta}_{\mu+1}, \ldots, \tilde{\beta}_n) \leq \B \cap (\Ker T_2)^{\perp}.
 \end{equation*}
 Moreover, $\tilde{\beta}_{\mu+1}, \ldots, \tilde{\beta}_n$ are linearly independent, since $S$ is regular. 
 For proving that the spaces are equal, it suffices to show that $\dim (\B \cap (\Ker T_2)^{\perp})= n - \mu$.
 Since $\B$ is finite dimensional, also $\dim (\B \cap (\Ker T_2)^{\perp}) < \infty$, and we have
 \begin{equation*}
  \dim (\B \cap (\Ker T_2)^{\perp}) = \codim(\B \cap (\Ker T_2)^{\perp})^{\perp} 
  = \codim (\Bo + \Ker T_2), 
 \end{equation*}
 where we use the duality principle for switching between a vector space and its dual 
 as explained in \cite[Sec.~A.1]{RegensburgerRosenkranz2009}. Furthermore, \cite[Cor.~A.15]{RegensburgerRosenkranz2009}
 yields
 \begin{equation*}
  \codim(\B \cap (\Ker T_2)^{\perp})^{\perp} = \dim (\Ker T_2 \cap \Bo ) + \codim \Bo - \dim \Ker T_2.
 \end{equation*}
 Since $(T_2, \B)$ is semi-regular, the first summand vanishes, and moreover using that $\codim \Bo = \dim \B$, we
 obtain $ \dim (\B \cap (\Ker T_2)^{\perp}) = n- \mu$.
\end{proof}

\begin{example} \label{ex:fac2}
 We demonstrate the steps of the algorithm in detail, continuing Example \ref{ex:fac1}.
 Obviously, a regular fundamental system of $\D^2$ is given by $(1, x)$ and a right inverse of $\D^2$
 is $\A^2$. The evaluation matrix
 $M$ and the transformation matrix $S$ are given as
 \begin{equation*}
  M= \begin{pmatrix}
   0 & 1 \\ 0 &0\\1 &1\\0 &1\\0&0
  \end{pmatrix} \quad \text{and} \quad
 S= \begin{pmatrix}
   -1 &0&1&0&0\\1&0&0&0&0\\0&1&0&0&0\\-1&0&0&1&0\\0&0&0&0&1
  \end{pmatrix}.
 \end{equation*}
 Hence we compute the $\tilde{\beta}$ as follows
 \begin{equation*}
  \tilde{\beta}_1 = \E_1 - \E_0\D, \;
  \tilde{\beta}_2 = \E_0\D,\;
  \tilde{\beta}_3 = \E_0 \D^3,\;
  \tilde{\beta}_4 = \E_1\D - \E_0\D,\;
  \tilde{\beta}_5 = \E_1\D^3.
 \end{equation*}
 The space of boundary conditions for the right-hand factor is spanned by $\tilde{\beta}_1$ and $\tilde{\beta}_2$, which means 
 that $\B_2 = \Ll(\E_1, \; \E_0\D)$. We multiply $\tilde{\beta}_3, \tilde{\beta}_4$ and $\tilde{\beta}_5$ by $\A^2$ to
 obtain a basis of $\B_1$. Since $\D\A=1$ and $\E_0\A=0$, this yields
 $\B_1 = \Ll(\E_0\D, \E_1\A, \E_1\D)$, as already seen in Example \ref{ex:fac1}.
 \end{example}
 
 The previous example shows that the computation of the boundary conditions $\B_1$ and $\B_2$ reduces to a linear algebra
 problem, which in our implementation is solved by the \maple\ procedures from the \texttt{LinearAlgebra}
 package. In the \pack\ package, the most time-consuming operations---factoring the differential operator and computing 
 a fundamental system for the right factor---are done in preprocessing steps via the \maple\ commands \texttt{dsolve} and
 \texttt{DFactor}.
  
 \begin{example}\label{ex:RatCoeff}
  
  On the interval $[0, 1]$, we consider the differential operator 
  \begin{multline*}
   \D^4 + \frac{5 x^2+4 x + 1}{(x+1)(x^2+1)} \D^3 + \frac{x^7+x^6+2 x^5+2 x^4 -x^3-5x^2+14x+10}{(x+1)(x^2+1)^2} \D^2\\
   + \frac{2(2x^8+2x^7+4x^6+4x^5 + x^4+2x^3-14x^2-16x+3)}{(x^2+1)^3(x+1)} \D \\
   + \frac{2(x^7+x^6+2x^5+2x^4 +5x^3+7x^2-4x-2)}{(x^2+1)^3(x+1)}
  \end{multline*}
  with coefficients in $\Q(x)$.
  
  \begin{center}
\begin{minipage}{0.999\textwidth}
\begin{framed}~\\[-1cm]\small
\begin{alltt} \negsskip \scriptsize
 > a3 := (5*x^2+4*x+1)/((x+1)*(x^2+1)):
 > a2 := (x^7 + x^6 + 2*x^5 + 2*x^4-x^3-5*x^2+14*x+10)/((x+1)*(x^2+1)^2):
 > a1 := 2*(2*x^8+2*x^7+4*x^6+4*x^5+x^4+2*x^3-14*x^2-16*x+3)/((x^2+1)^3*(x+1)):
 > a0 := 2*(x^7+x^6+2*x^5+2*x^4+5*x^3+7*x^2-4*x-2)/((x^2+1)^3*(x+1)):
 > t := d^4 + a3.(d^3) + a2.(d^2) + a1.d + a0:
 > b1 := e(0): b2 := e(0).d: b3 :=e(0).(d^2): b4 := e(1): b5 := e(1).d: \small
 > bp := GBP(t, BC(b1, b2, b3, b4, b5), ES(1)):
 > FactorBoundaryProblem(bp);
\end{alltt}  \scriptsize 
\begin{align*}
&\GBP(x^2+\frac{1}{1+x} . \D+\D^2, \BC(\E[0], \E[1] . \A . x^2+\E[1] . \A, \E[1] . \A . x^3+\E[1] . \A . x), \ES(1)),\\
&\BP(\frac{2x}{x^2+1}+\D, \BC(\E[0])),\;\;\BP(\frac{2x}{x^2+1}+\D, \BC(\E[0]))
% 
%  \GBP(\D^4-\D^2, \BC(\E[0].\D, \E[0].\D^3, \E[1], \E[1].\D, \E[1].\D^3), \ES(x))
\end{align*} 
\end{framed}\normalsize
\end{minipage}
\end{center}
  
 \end{example}
 
 We conclude this section with an example for a more general differential operator, which cannot be factored
 with the \texttt{DFactor} command. In this case, a factorization of the differential operator has to be provided
 in the input data.
 
 \begin{example} \label{ex:exp}
 We consider a differential operator on $\Omega=(0, \infty)$ with coefficients in $C^{\infty}(\Omega)$.
 \begin{equation*}
  \D^2 - \frac{\exp(x) + \exp(2x) -1}{\exp(x)-1} \D + \frac{\exp(2x)}{\exp(x)-1}
 \end{equation*}
with boundary conditions $\E_1, \E_2, \E_3$, and the factorization 
\begin{equation*}
 T_1 = \D- \frac{\exp(2x)}{\exp(x)-1}\quad \text{and} \quad T_2 = \D-1.
\end{equation*}

   \begin{center}
\begin{minipage}{0.999\textwidth}
\begin{framed}~\\[-1cm]\small
\begin{alltt} 
 > t := d^2-(exp(x)+exp(2*x)-1)/(exp(x)-1).d+exp(2*x)/(exp(x)-1):
 > t1 := d-exp(2*x)/(exp(x)-1):
 > t2 := d-1:
 > b1 := e(1): b2 := e(2): b3 := e(3):
 > bp := GBP(t, BC(b1, b2, b3), ES(1)):
 > FactorBoundaryProblem(bp, t1, t2);
\end{alltt}  \scriptsize 
\begin{align*}
&\GBP(\frac{-\exp(2x)}{\exp(x)-1} +\D, \BC(\E[1].\A.\exp(-x), \E[2].\A.\exp(-x) -\E[3].\A.\exp(-x) ), \ES(1)),\\
&\BP(-1+\D, \BC(\E[1]))
% 
%  \GBP(\D^4-\D^2, \BC(\E[0].\D, \E[0].\D^3, \E[1], \E[1].\D, \E[1].\D^3), \ES(x))
\end{align*} 
\end{framed}\normalsize
\end{minipage}
\end{center}

 \end{example}

\section{Conclusion and Outlook}
\label{sec:outlook}
 
As outlined in the previous section, describing all possible factorizations of a generalized boundary problem along a fixed factorization of the differential operator is quite involved. 
 Nevertheless, it would sometimes be preferable to have the generalized factor on the right, 
for which we assume symbolic solutions for the differential operator.
In this section, we describe some first steps in this direction that rely on the following result.
 
 \begin{theorem} \label{thm:ChoseE2}
Let $(T_1, \B_1)$ and $(T_2, \B_2)$ be semi-regular with $T_1 \colon V \to W$ and $T_2 \colon U \to V$.
Then there exists an exceptional space $E_2 \leq V$ of $(T_2, \B_2)$ such that the 
respective generalized Green's operators satisfy the reverse order law for all possible 
exceptional spaces $E_1 \leq W$ of $(T_1, \B_1)$.
\end{theorem}

\begin{proof}
 Let $V_1$ be a complement of $\Bo_1 \cap T_2(\Bo_2)$ in $\Bo_1$, i.e.,
 \[\Bo_1 = \bigl(\Bo_1 \cap T_2(\Bo_2) \bigr) \dotplus V_1.\]
 Since $V_1 \leq \Bo_1$ and because of the direct sum, we then have
 \[V_1 \cap T_2(\Bo_2) 
  = \bigl( V_1 \cap T_2(\Bo_2) \bigr) \cap \Bo_1
  = V_1 \cap \bigl( T_2(\Bo_2)  \cap \Bo_1 \bigr) = \{0\}.
 \]
Enlarging $V_1$ to a complement $E_2$ of $T_2(\Bo_2)$, so that $V_1 \leq E_2$ and
$T_2(\Bo_2) \dotplus E_2 = V$ yields 
\[ \Bo_1 = \bigl(\Bo_1 \cap T_2(\Bo_2) \bigr) \dotplus V_1 \leq T_2(\Bo_2) + (E_2 \cap \Bo_1).\]
Hence Condition~\eqref{eq:OI3} is satisfied.
\end{proof}

The previous proof is not constructive; choosing exceptional spaces for arbitrary Stieltjes
boundary conditions leads to an interpolation problem with integral conditions. However,  one can apply the 
following strategy to obtain a factorization of a semi-regular boundary problem $(T, \B)$ into
a regular problem $(T_1, \B_1)$ as a left factor
and a generalized boundary problem as a right factor. In this case we will ignore the exceptional
space in the beginning and (try to) choose it accordingly in the end.

\begin{enumerate}
 \item Apply (a version of) Algorithm~\ref{alg:RightFact} to compute a semi-regular factor $(T_1, \B_1)$ and
  a regular factor $(T_2, \B_2)$.
 \item Let $\alpha_1, \ldots, \alpha_n$ be a basis of $\B_1$. Choose a subset (wlog we write $\alpha_{1}, \ldots, \alpha_{m}$), 
 such that $(T_1, \Ll(\alpha_{1}, \ldots, \alpha_{m}))$ is regular. Such a subset exists by Lemma~\ref{lem:intersections}
 and can also be computed.
 \item Compute $T_2^*(\alpha_{m+1}), \ldots, T_2^*(\alpha_{n})$.
 \item Since $(T_2, \B_2)$ is regular, $(T_2, \B_2 + \Ll(T_2^*(\alpha_{m+1}), \ldots, T_2^*(\alpha_{n}))$ is 
 semi-regular, and from Theorem \ref{thm:ChoseE2} we obtain an appropriate exceptional space.
\end{enumerate}

\begin{example}
 In Example \ref{ex:fac2}, we have computed the spaces of boundary conditions 
 \begin{equation*}
  \B_1 = \Ll(\E_0\D, \E_1 \D, \E_1\A) \quad \text{and} \quad
  \B_2 = \Ll(\E_0\D, \E_1)
 \end{equation*}
 for the differential operators $T_1 = \D^2 -1$ and $T_2 = \D^2$. The boundary problem $(T_1, \Ll\E_0\D, \E_1\D)$ is regular,
 and we compute 
 \begin{equation*}
   T_2^*(\E_1\A)= \E_1\A\D^2 = \E_1\D - \E_0\D.
 \end{equation*}
 Computation of the compatibility conditions for the modified right-hand factor
 $(T_2, \Ll(\E_0\D, \E_1, \E_1\D))$ yields $\C_2 = \Ll(\E_1\A)$, hence we may choose for example
 $E_2 = \Ll(\exp(x))$. Then, since $E_2 \leq \Ker T_1$, Condition~\ref{algOI2} of Theorem~\ref{thm:G2G1OI}
 is trivially fulfilled.
\end{example}

Closer investigation of Algorithm \ref{alg:RightFact} and the above procedure indicates some redundancy in first applying
$H_2^*$ and afterwards $T_2^*$. However, we cannot decide in advance for which choices of $\alpha_i$ the boundary problem
$(T_1, \Ll(\alpha_1, \ldots, \alpha_m))$ will be regular. We will study this and related question for obtaining different factorizations of generalized boundary problems in future work.

We are also investigating algebraic and algorithmic aspects of generalized Green's  operators for certain classes of linear partial differential equations and boundary conditions. While the linear algebra setting used in this paper is in principle applicable to partial differential equations, examples and constructive methods for singular partial boundary problems are to be studied. In this context, we refer to~\cite{RegensburgerRosenkranz2009,RosenkranzRegensburgerTecBuchberger2009,Tec2011,RosenkranzPhisanbut2013}
for regular partial boundary problems. The setting in~\cite{RosenkranzPhisanbut2013} includes also  inhomogeneous boundary conditions and a rewrite system for partial integro-differential operators (PIDOS) including linear substitution.

\bigskip

\noindent \textbf{Acknowledgements.} We would like to thank the anonymous referees for their numerous helpful comments. G.R.~was partially supported by the Austrian Science Fund (FWF): J 3030-N18.

%\bibliographystyle{splncs}
%\bibliography{ACA_cited}

\end{document}